\documentclass[nonacm,acmsmall]{acmart}

\usepackage{enumitem}
\usepackage{tikz}
\usetikzlibrary{
  arrows.meta,
  positioning,
  calc,
  backgrounds,
  decorations.pathreplacing
}
\definecolor{geInk}{HTML}{20334A}
\definecolor{geBlue}{HTML}{2E668F}
\definecolor{geTeal}{HTML}{1B8A7A}
\definecolor{geGold}{HTML}{D9A128}
\definecolor{geCoral}{HTML}{D66853}
\definecolor{geMist}{HTML}{EAF1F5}
\definecolor{geGray}{HTML}{8794A5}

\tikzset{
  figure header/.style={
    font=\bfseries\large,
    text=black,
    align=center
  }
}

\tikzset{
  ge edge/.style={draw=geInk!68,line width=.8pt,line cap=round},
  ge faint edge/.style={draw=geInk!22,line width=.65pt,line cap=round},
  ge strong edge/.style={draw=geBlue,line width=1.8pt,line cap=round},
  ge arrow/.style={-Latex,draw=geCoral,line width=1.5pt},
  ge vertex/.style={circle,fill=geInk,draw=white,line width=.6pt,
                    minimum size=6pt,inner sep=0pt},
  ge blue vertex/.style={circle,fill=geBlue,draw=white,line width=.7pt,
                         minimum size=8pt,inner sep=0pt},
  ge teal vertex/.style={circle,fill=geTeal,draw=white,line width=.7pt,
                         minimum size=8pt,inner sep=0pt},
  ge gold vertex/.style={circle,fill=geGold,draw=white,line width=.7pt,
                         minimum size=8pt,inner sep=0pt},
  ge coral vertex/.style={circle,fill=geCoral,draw=white,line width=.7pt,
                          minimum size=8pt,inner sep=0pt},
  ge label/.style={font=\small,fill=white,inner sep=1.5pt,text=geInk},
  ge tiny label/.style={font=\scriptsize,fill=white,inner sep=1pt,text=geInk}
}

\setcopyright{none}
\newcommand{\eps}{\varepsilon}
\newcommand{\R}{\mathbb R}

\newcommand{\AN}{\mathsf{AN}}
\newcommand{\Zar}{\mathsf Z}
\newcommand{\ex}{\operatorname{ex}}
\newcommand{\codeg}{\operatorname{codeg}}
\newcommand{\outdeg}{\deg^{+}}
\title{The Price of Almost Navigability}

\author{Tomer Waizer}
\orcid{0009-0000-7589-5257}
\affiliation{%
  \department{Taub Faculty of Computer Science}
  \institution{Technion -- Israel Institute of Technology}
  \city{Haifa}
  \country{Israel}
}
\email{tomer.waizer@campus.technion.ac.il}

\author{Yoav Danieli}
\orcid{0000-0001-8774-064X}
\affiliation{%
  \department{Taub Faculty of Computer Science}
  \institution{Technion -- Israel Institute of Technology}
  \city{Haifa}
  \country{Israel}
}
\email{yoavd@campus.technion.ac.il}

\begin{document}

\begin{abstract}

\textit{Navigability} is a fundamental property of graph-based search structures and plays an important role in the analysis of nearest-neighbor algorithms. Informally, a graph is navigable if, from any current point and toward any desired target, there is always an outgoing edge that moves strictly closer to that target. While this property provides a strong guarantee for greedy search, it can be inherently expensive: in the worst case, navigable graphs require $\Omega(n^{3/2})$ edges, where $n$ is the size of the dataset.

Recently, Avi and Musco introduced $(1-\eps)$-\textit{almost navigability}, a natural relaxation in which, from every current point, such a progress-making edge is required for almost all targets, while an $\eps$ fraction of targets may fail this condition \cite{avimusco2026almost}. They showed that every dataset admits such a graph with $O(n/\eps)$ edges.

In this work, we prove a matching lower bound, establishing the optimality of their construction and giving a complete characterization of the sparsity achievable by almost-navigable graphs. For most values of $\eps$, our hard instances lie in Euclidean spaces of polylogarithmic dimension, across the full worst-case range of $\eps$, dimension \(d=O(\sqrt{n}\log^{3/2} n)\) suffices. The same construction also sharpens our understanding of ordinary navigability: in the latter dimension, we exhibit datasets for which every navigable graph has $\Omega(n^{3/2})$ edges.

Our proof reveals a surprising connection between almost navigability and the classical Zarankiewicz problem of constructing dense graphs with limited pairwise neighborhood overlap. This connection lets us translate extremal graph constructions into hard geometric instances for navigation, linking two seemingly different notions of graph sparsity.
\end{abstract}

\ccsdesc[500]{Information systems~Proximity search}
\keywords{nearest-neighbor search, navigability, dimension reduction, Zarankiewicz problem}

\maketitle

\section{Introduction}
\label{sec:introduction}

Nearest-neighbor search is a fundamental and widely studied problem in data analysis, machine learning, and information retrieval.
Given a dataset
$P=\{p_1,\ldots,p_n\}\subseteq\mathbb R^d$ and a query point $q\in\mathbb R^d$,
one would like to find the point of $P$ closest to $q$.
In
high dimension, exact search is subject to the familiar \emph{curse of
dimensionality}: classical indexing methods deteriorate with dimension, and
general lower bounds show strong dimension dependence
~\cite{indykMotwaniANN,weberSchekBlott1998,borodinOstrovskyRabani1999,barkolRabani2002}.
Therefore, much of the literature instead
studies approximate nearest-neighbor (ANN) search.

Graph-based methods have become especially successful in practice.  HNSW,
DiskANN, NSG, and related systems are among the strongest methods on
large-scale ANN benchmarks
~\cite{aumullerANNBenchmarks,simhadriBigANN2021,simhadriBigANN2023,
malkov2020hnsw,subramanya2019diskann,fu2019nsg}.  Although their construction
rules differ, their search procedures share a local principle: from the
current point, follow edges to points that are closer to the query.

This motivates studying properties of the graph itself, including its
size~\cite{lutao2026proximity}. A particularly clean one is
\emph{navigability}, which is related to earlier work on decentralized
routing and small-world networks
~\cite{milgramSmallWorld,traversMilgram,kleinberg2000smallworld,bogunaNavigability}
and has recently become a useful abstraction in the theory of graph-based
nearest-neighbor search
~\cite{indyk2023worst,diwan2024navigable,khanna2026sparse,conway2026navigable,DiCicco0025}.

Navigability guarantees that greedy search cannot get stuck on an in-dataset
target: unless the target has been reached, some outgoing edge makes strict
progress toward it. The guarantee is strong, but it can require many edges.
Conway et al.~\cite{conway2026navigable} show that every $n$-point dataset
admits a navigable graph with $O(n\sqrt{n})$ edges and $O(\sqrt{n})$ average out-degree.
Diwan et al.~\cite{diwan2024navigable} proved lower bounds of
$\Omega(\sqrt{n}/\log n)$ average out-degree, even for Euclidean
datasets in $\Omega(\log^3 n)$ dimensions. Thus their result gets arbitrarily close
to, but does not reach, the $\sqrt n$ scale.

Motivated by the need for substantially sparser constructions in practical
settings, Avi and Musco~\cite{avimusco2026almost} proposed weakening full
navigability by permitting a small fraction of targets to fail.

In the same paper, Avi and Musco
proved that every dataset admits a
$(1-\eps)$-almost navigable graph with average out-degree at most $4/\eps$.
Thus the natural worst-case scales are $1/\eps$ for almost navigability and
$\sqrt n$ for full navigability.

We write $\AN(n,\eps)$ for the worst-case Euclidean cost, measured by
average out-degree (Definition~\ref{def:worst-case-cost}).
Our results show that both scales in~\eqref{eq:known-upper} are unavoidable
and determine their complete transition.

\begin{enumerate}[label=\textbf{(\arabic*)}]
\item \textbf{A spherical construction.}
For $C\sqrt{\log n/n}\leq\eps\leq\eps_0$, an i.i.d.
sample from a sphere in
$O(\log n\,\log^2(1/\eps))$ dimensions forces minimum out-degree
$\Omega(1/\eps)$ with high probability.
As a corollary we obtain that
the construction of Avi and Musco is optimal and the
lower bound holds at every source.

\item \textbf{Connection to the Zarankiewicz problem.}
We reveal a direct connection between almost navigability and a classical problem in extremal graph theory, the Zarankiewicz problem.
This problem asks how many edges a graph can have while avoiding a prescribed complete bipartite subgraph, see more in Section~\ref{sec:zarankiewicz}.
We show that dense graphs avoiding $K_{2,t+1}$ give rise to hard instances for navigation, and, conversely, that any universal upper bound for almost navigability implies an upper bound on the maximum density of such graphs.
In particular, applying the known $O(1/\eps)$ construction for almost navigability recovers the classical $O(\sqrt{t}n^{3/2})$ bound for $K_{2,t+1}$-free graphs~\cite{kovari1954zarankiewicz}.

\item \textbf{A F\"uredi-graph construction.}
Using F\"uredi's extremal $K_{2,t+1}$-free graphs, we construct generic
Euclidean datasets for which almost navigability requires large average
out-degree.
Specifically, for the complement regime
$0\leq \eps\leq C\sqrt{\log n/n}$, our construction gives

$$
  \AN(n,\eps)=\Omega\left(\min\left\{\frac1\eps,\sqrt n\right\}\right)
$$

The datasets can be realized in dimension
$O(\sqrt n\log^{3/2}n)$ throughout this regime, and in dimension
$O(\sqrt n\log n)$ when $\eps=0$.  In particular, at full navigability there
are generic Euclidean datasets on which every navigable graph has average
out-degree $\Omega(\sqrt n)$, or equivalently $\Omega(n^{3/2})$ edges.
\end{enumerate}

The rest of the paper follows this dependency order.
Section~\ref{sec:preliminaries} collects the definitions for navigation,
charging, and the two Euclidean constructions and records the known upper bound.
Section~\ref{sec:charging}
introduces a combinatorial machinery
known as the 'global charging theorem'
and its proof.
Section~\ref{sec:sphere} gives the
spherical lower bound.  Section~\ref{sec:zarankiewicz} establishes the
Zarankiewicz transference theorem, after which
Section~\ref{sec:furedi} applies F\"uredi's construction
and derives the complete phase.  Further technical details and additional
sourcewise constructions are given in Appendix~\ref{app:lower-bounds}.

\section{Preliminaries}
\label{sec:preliminaries}

\subsection{Approximate nearest-neighbor search}

\begin{definition}
For $c\geq1$, a $c$-approximate nearest neighbor of a query
$q\in\mathbb R^d$ is a point $p\in P$ satisfying

$$
        \|p-q\|_2
        \leq
        c\min_{x\in P}\|x-q\|_2.
$$

\end{definition}

Several approaches have been developed for this problem. Locality-sensitive
hashing uses randomized hash functions that favor collisions between nearby
points~\cite{indykMotwaniANN,andoniIndykLSH,andoniPracticalLSH,lvMultiProbeLSH}.
Quantization and clustering reduce the search space by compressing or
partitioning the data~\cite{jegouProductQuantization,johnsonBillionScale,douzeFaiss},
while hierarchical methods adapt tree-based search to high-dimensional
settings~\cite{bawaLSHForest,beygelzimerCoverTrees,andoniLSHForest}.

\subsection{Navigation model}

We work slightly more generally than Euclidean distance.

\begin{definition}
\label{def:progress}
Let $P$ be a set of $n$ points equipped with a symmetric distance function
$d:P\times P\to\R_{\geq0}$ satisfying $d(x,x)=0$, $d(x,y)=d(y,x)$, and $d(x,y)>0$ for $x\neq y$.
Moreover, let $G=(P,E)$ be a
directed graph without self-loops.  For a source $u\in P$ and a target
$t\in P\setminus\{u\}$, an edge $(u,v)\in E$ \emph{makes progress toward
$t$} if $d(v,t)<d(u,t)$.  We write
\[
 C_G(u)=\{t\in P\setminus\{u\}:\exists (u,v)\in E
                 \text{ with }d(v,t)<d(u,t)\}
\]
for the targets toward which $u$ has a progress-making outgoing edge.
\end{definition}

\begin{definition}
A directed graph $G=(P,E)$ is \emph{navigable} if

$$
                         C_G(u)=P\setminus\{u\}
                         \qquad\text{for every }u\in P.
$$

\end{definition}

\begin{definition}
\label{def:almost-nav}
For $\eps\in[0,1]$, the graph $G$ is $(1-\eps)$-\emph{almost navigable} if
\[
                       |C_G(u)|\geq(1-\eps)(n-1)
                       \qquad\text{for every }u\in P.
\]
The case $\eps=0$ is \emph{navigability}.
\end{definition}

We focus on Euclidean datasets and use genericity only to rule out ties in
distance orderings.

\begin{definition}
\label{def:worst-case-cost}
A finite set $P\subset\R^D$ is \emph{generic} if all pairwise distances
between distinct points are different.

For a generic $n$-point Euclidean dataset $P$, consider the sparsest
$(1-\eps)$-almost navigable graph on $P$, measured by its average out-degree.
We define $\AN(n,\eps)$ to be the worst possible value of this quantity over
all generic $n$-point Euclidean datasets, in any dimension.
\end{definition}

Combining the constructions of Avi--Musco and Conway et al. gives
\begin{equation}
 \AN(n,\eps)\leq
 \min\left\{\frac4\eps,2\sqrt n\right\}.
 \label{eq:known-upper}
\end{equation}

\subsection{Distance down-sets and pair load}
\label{sec:prelim-downsets}

All of our lower bounds are based on the same counting principle.  The idea is
to identify many source-target pairs that require a navigation edge, and then
bound how many such pairs a single edge can serve.

For each target $t\in P$, we choose a set $B_t\subseteq P$ of points that are
``close'' to $t$.  The precise meaning of close will depend on the
application.  The only property we require is that $B_t$ contain every point
that is closer to $t$ than one of its members.  Thus, once a source $u$ lies
in $B_t$, every move from $u$ that makes progress toward $t$ must remain
inside $B_t$.

\begin{definition}
\label{def:downset}
A family $\mathcal B=(B_t)_{t\in P}$ is \emph{downward closed} if $t\in B_t$
for every $t\in P$, and

$$
 u\in B_t,\qquad d(v,t)<d(u,t)
 \quad\Longrightarrow\quad
 v\in B_t.
$$

For such a family, define its \emph{pair load} by

$$
 \Lambda(\mathcal B)
 =
 \max_{u\ne v}
 |\{t\in P:u,v\in B_t\}|.
$$

In other words, $\Lambda(\mathcal B)$ is the largest number of sets $B_t$
that can simultaneously contain the same pair of points.
\end{definition}

The pair load measures how efficiently one directed edge can be reused.
Indeed, suppose that $u\in B_t$ and an edge $u\to v$ makes progress toward
$t$.  Since $v$ is closer to $t$ than $u$, downward closure implies that
$v\in B_t$ as well.  Thus $u\to v$ can help with target $t$ only if both of
its endpoints lie in $B_t$.  Consequently, any fixed edge can help at most
$\Lambda(\mathcal B)$ different targets.

We write $x_+=\max\{x,0\}$ for the positive part of $x$ and
$\outdeg_G(u)$ for the out-degree of $u$ in $G$.

\subsection{Spherical caps}

In the spherical construction, the sets $B_t$ come from spherical caps.
We choose the threshold by fixing the measure of each cap.

Let $X$ be uniform on the unit sphere $\mathbb S^{d-1}$.  For a unit vector
$u$ and $a\geq0$, put
\begin{equation}
       \operatorname{Cap}(u,a)
       =\left\{x\in\mathbb S^{d-1}:
                    \sqrt d\,\langle x,u\rangle\geq a\right\},
 \qquad
       p_d(a)=\Pr(X\in\operatorname{Cap}(u,a)).
\end{equation}
By rotation invariance, $p_d(a)$ does not depend on $u$.

\subsection{Codegree and the Zarankiewicz problem}
\label{sec:prelim-codegree}

The Zarankiewicz problem asks how many edges a graph can have while avoiding a
fixed complete bipartite graph $K_{s,t}$~\cite{kovari1954zarankiewicz,
furedi1996asymptotics}.  We will use the case $K_{2,t+1}$.  A graph is
$K_{2,t+1}$-free exactly when no two vertices have more than $t$ common
neighbors: indeed, two vertices together with $t+1$ of their common neighbors
form a copy of $K_{2,t+1}$.  Thus, in this case, the Zarankiewicz problem asks
how dense a graph can be while keeping the overlap between every two
neighborhoods small.

For a vertex $u$ of a simple graph $H$, write $N_H(u)$ for its open
neighborhood and $N_H[u]=N_H(u)\cup\{u\}$ for its closed neighborhood.

For a simple graph $H$, write

$$
 \codeg(H)
 =
 \max_{u\neq v}|N_H(u)\cap N_H(v)|
$$

for its maximum codegree.  We also write

$$
 \Zar(n,t)
 =
 \ex(n,K_{2,t+1})
 =
 \max\bigl\{
 |E(H)|:\ |V(H)|=n,\ \codeg(H)\leq t
 \bigr\}.
$$

Thus $\Zar(n,t)$ is the largest number of edges in an $n$-vertex graph in
which every two vertices have at most $t$ common neighbors.

A dense graph has many vertex--neighborhood incidences, while bounded common
neighborhoods limit how often a pair of vertices can occur together. If the
neighborhoods of such a graph can be realized as metric balls in Euclidean
space, they can therefore serve as the sets $B_t$ in
Theorem~\ref{thm:global-charging}.

\section{The global charging theorem}
\label{sec:charging}

For the near-target pairs in Definition~\ref{def:downset}, the pair load
bounds how many targets a single edge can serve.
This immediately gives our main counting lemma.

\begin{theorem}
\label{thm:global-charging}
For every downward-closed family $\mathcal B$ and every
$(1-\eps)$-almost navigable graph $G$,
\begin{equation}
|E(G)|
\geq
\frac{
\left(
\sum_{t\in P}(|B_t|-1)-\eps n(n-1)
\right)_+
}
{\max\left({1,\Lambda(\mathcal B)}\right)}.
\label{eq:global-charging-main}
\end{equation}
\end{theorem}

\begin{proof}
Consider all ordered pairs $(u,t)$ with $u\neq t$ and $u\in B_t$.  The number of such pairs is exactly

$$
 \sum_{t\in P}(|B_t|-1)\, .
$$

For each such pair, navigability asks for an outgoing edge $u\to v$ satisfying
$d(v,t)<d(u,t)$.  Since $G$ is only $(1-\eps)$-almost navigable, up to
$\eps n(n-1)$ ordered source--target pairs may fail to have such an edge.
Therefore at least

$$
 \left(
 \sum_{t\in P}(|B_t|-1)-\eps n(n-1)
 \right)_+
$$

of the pairs under consideration admit an improving edge.

Choose one improving edge for each surviving pair $(u,t)$ and charge the pair
to that edge.  If $(u,t)$ is charged to $u\to v$, then $u\in B_t$ and
$d(v,t)<d(u,t)$.  By downward closure, $v\in B_t$.  Hence the edge $u\to v$
can receive charges only from targets $t$ for which both $u$ and $v$ belong
to $B_t$, and there are at most $\Lambda(\mathcal B)$ such targets.

Thus each edge receives at most $\Lambda(\mathcal B)$ charges.  Comparing the
number of surviving source--target pairs with the number of edges proves
\eqref{eq:global-charging-main}.
\end{proof}

To obtain a strong lower bound, we therefore seek families for which
the sets $B_t$ are large but their pairwise overlap, as measured by
$\Lambda(\mathcal B)$, is small.

\subsection{Local common-ball charging}
\label{app:proof-local-common-ball}

A particularly useful case occurs when all sets are metric balls of the same
radius.  Then the same charging argument can be applied separately at each
source, giving a pointwise lower bound on its out-degree.

\begin{corollary}
\label{cor:local-common-ball}
Suppose that for some radius $r$,

$$
 B_t=\{x\in P:d(x,t)\leq r\}
 \qquad\text{for every }t\in P.
$$

Then every $(1-\eps)$-almost navigable graph satisfies
\begin{equation}
\outdeg_G(u)
\geq
\frac{
\bigl(|B_u|-1-\eps(n-1)\bigr)_+
}
{\max\left({1,\max_{v\ne u}|B_u\cap B_v|}\right)}
\qquad\text{for every }u\in P.
\label{eq:local-common-ball-main}
\end{equation}
\end{corollary}

\begin{proof}
Fix a source \(u\in P\).
By symmetry of the common-radius balls,
\[
    u\in B_t
    \quad\Longleftrightarrow\quad
    t\in B_u.
\]
Hence there are exactly \(|B_u|-1\) targets \(t\neq u\) for which
\(u\in B_t\).

Since \(G\) is \((1-\eps)\)-almost navigable, at most
\(\eps(n-1)\) of these targets can fail to have an improving
outgoing edge from \(u\).  Thus at least
\[
    |B_u|-1-\eps(n-1)
\]
targets must be served by some edge \((u,v)\).

Fix such an edge \((u,v)\).  If it makes progress toward a target
\(t\), then \(u\in B_t\) and
\[
    d(v,t)<d(u,t).
\]
Since \(B_t\) is downward closed, \(v\in B_t\) as well.
By symmetry,
\[
    u,v\in B_t
    \quad\Longleftrightarrow\quad
    t\in B_u\cap B_v.
\]
Therefore the edge \((u,v)\) can serve at most
\(|B_u\cap B_v|\) of the relevant targets.

Since \(u\) has only \(\outdeg_G(u)\) outgoing edges,
\[
    |B_u|-1-\eps(n-1)
    \le
    \outdeg_G(u)\,
    \max_{v\neq u}|B_u\cap B_v|.
\]
Rearranging, and inserting the positive part and the harmless
maximum with \(1\), gives the claim.
\end{proof}

\section{The spherical construction}
\label{sec:sphere}

We first realize the charging principle directly in Euclidean space using a
simple random construction.  Sample $n$ points independently and uniformly
from the unit sphere $\mathbb S^{d-1}$, and around each sample point $t$ take
a spherical cap containing an $\Theta(\eps)$ fraction of the sphere.

The reason this construction is useful is that spherical caps behave, in high
dimension, almost as though membership in two different caps were independent.
A cap of spherical measure $p$ contains about $pn$ sample points, whereas two
caps centered at distinct sample points typically contain only about $p^2n$
sample points in common.  Comparing this with
Corollary~\ref{cor:local-common-ball}, the first quantity determines how many
nearby targets a source must serve, while the second controls how many of
those targets can reuse the same outgoing edge.  Taking $p=\Theta(\eps)$
therefore suggests an out-degree lower bound of order

$$
 \frac{pn}{p^2n}=\frac1p=\Theta\!\left(\frac1\eps\right).
$$

The main work is to show that these estimates hold simultaneously for all
sample points and all pairs of centers.

\begin{theorem}
\label{thm:sphere-main}
There are universal constants $c,C,\eps_0>0$ such that the following holds.
Suppose

$$
 C\sqrt{\frac{\log n}{n}}\leq\eps\leq\eps_0,
 \qquad
 d\geq C\log n\,\log^2(1/\eps).
$$

Let $P$ consist of $n$ independent points sampled uniformly from
$\mathbb S^{d-1}$.  With probability at least $1-n^{-2}$, every
$(1-\eps)$-almost navigable graph $G$ on $P$ satisfies

$$
 \min_{u\in P}\outdeg_G(u)\geq\frac{c}{\eps},
 \qquad
 |E(G)|\geq\frac{cn}{\eps}.
$$

\end{theorem}

The theorem shows that the known $O(1/\eps)$ upper bound is already optimal,
up to constants, on random Euclidean datasets throughout this range of
$\eps$.  Moreover, the conclusion is stronger than a bound on the total
number of edges: every vertex must itself have out-degree
$\Omega(1/\eps)$.
\paragraph{Proof overview.} Choose a cap threshold
whose mass is $p=8\eps$, and let $B_t$ be the sample points in the cap centered at $t$.
Uniform Chernoff bounds give
$|B_t|-1\geq p(n-1)/2$ for every center.
The main geometric estimate shows that, if two centers have correlation $\rho$ and the normalized cap threshold is $a$,
then \[ \Pr(X\text{ lies in both caps}) \leq Cp^2\exp\!\left(C|\rho|a^2+C\frac{a^4}{d}\right). \] Here $a^2=O(\log(1/p))$, while the maximum correlation among the sample points is $O(\bigl(\log n/d\bigr)^{1/2})$. The stated dimension makes the exponential correction constant. A union bound then yields $|B_u\cap B_v|=O(p^2n)$ for all pairs, and Corollary~\ref{cor:local-common-ball} gives $\outdeg_G(u)=\Omega(1/p)$ at every source. Appendix~\ref{app:sphere-proof} contains the relative two-cap estimate.

\subsection{Uniform incidence and the stronger parameterized theorem}

The following parameterized statement is slightly stronger than
Theorem~\ref{thm:sphere-main} and is convenient for the proof.

\begin{theorem}
\label{thm:sphere-parameterized}
There are universal constants $c,C,{A>0}$, $n_0\geq1$, and $\eps_0>0$ such
that the following holds.  Let $n\geq n_0$, $0\leq\eps\leq\eps_0$, put
\[
 p=\max\left\{8\eps, A\sqrt{\frac{\log n}{n}}\right\},
 \qquad
 d\geq C\log n\,\log^2(1/p)\, .
\]
With probability at least
$1-n^{-2}$, every $(1-\eps)$-almost navigable graph on $n$ i.i.d. uniform points
of $\mathbb S^{d-1}$ satisfies
\[
 \min_u\outdeg_G(u)\geq\frac{c}{p}
 =\Omega\!\left(\min\left\{\frac1\eps,
                     \sqrt{\frac n{\log n}}\right\}\right),
 \qquad |E(G)|=\Omega(n/p).
\]
\end{theorem}

The proof has two ingredients: a uniform relative estimate for the
intersection of two spherical caps and the local charging inequality from
Corollary~\ref{cor:local-common-ball}.

Fix universal constants $p_0\in(0,1/20]$, $A\geq1$, and $C_2\geq1$ as in the proof of Lemma~\ref{lem:uniform-incidence-codegree}.  For all sufficiently large $n$, assume
\begin{equation}
 A\sqrt{\frac{\log n}{n}}\leq p\leq p_0,
 \qquad
 d\geq C_2\log n\,\log^2(1/p).
\end{equation}
Choose the unique $a>0$ for which $p_d(a)=p$.  Taking $p_0$ sufficiently
small ensures $a\geq1$: for $n_0$ large we have $d\geq16$, and
Lemma~\ref{lem:sphere-marginal} at $a=1$ gives a universal positive lower
bound on $p_d(1)$.  By the global spherical subgaussian estimate in the same
lemma, with $C_5=c_0^{-1}$,
\begin{equation}
                       a^2\leq C_5\log(1/p).
\end{equation}
After increasing $C_2$, this upper bound on $a$ also ensures
$a\leq\sqrt d/4$.

Draw $x_1,\ldots,x_n$ independently and uniformly from
$\mathbb S^{d-1}$.  For every center $x_t$, define the common-radius cap
neighborhood
\begin{equation}
 B_t=\{x_u:\sqrt d\,\langle x_u,x_t\rangle\geq a\}.
\end{equation}
Since $\|x_u-x_t\|_2^2=2-2\langle x_u,x_t\rangle$, the sets $B_t$ are
downward-closed metric balls.  They contain their centers.
By symmetry, the number of targets $t\ne u$ whose ball contains $u$ is
$D_u:=|B_u|-1$, for $u\ne v$, put $C_{uv}:=|B_u\cap B_v|$.  Thus
Corollary~\ref{cor:local-common-ball} applies with these parameters.

\begin{lemma}
\label{lem:uniform-incidence-codegree}
Under the parameter range above, the constants $A,C_2$ can be chosen
universally so that, with probability at least $1-n^{-2}$, simultaneously for
all distinct $u,v$,
\begin{equation}
             D_u\geq\frac{p(n-1)}{2},
             \qquad
             C_{uv}\leq C_3p^2n,
\end{equation}
for a universal $C_3$.
\end{lemma}

The proof of Lemma~\ref{lem:uniform-incidence-codegree} is given in
Appendix~\ref{app:sphere-incidence}.

\subsubsection{Proofs of the sphere lower bounds}

\begin{proof}[Proof of Theorem~\ref{thm:sphere-parameterized}]
Work on the event from Lemma~\ref{lem:uniform-incidence-codegree}.  For
$n\geq n_0$,
\[
 D_u-\varepsilon(n-1)
 \geq\left(\frac p2-\varepsilon\right)(n-1)
 \geq\frac{3p}{8}(n-1),
\]
where the last step uses $p\geq8\varepsilon$.
Corollary~\ref{cor:local-common-ball}
and $\max_{v\ne u}C_{uv}\leq C_3p^2n$ therefore give
\[
 \deg_G^+(u)
 \geq\frac{3p(n-1)}{8C_3p^2n}
 \geq\frac{1}{4C_3p}
\]
for every $u$.  This proves the minimum-outdegree bound, summing it over $u$
proves the edge bound.
\end{proof}

\begin{proof}[Proof of Theorem~\ref{thm:sphere-main}]
\begin{sloppypar}
Choose the constants so that the lower bound on $\varepsilon$ implies
$8\varepsilon\geq A\sqrt{\log n/n}$ and the upper bound implies
$8\varepsilon\leq p_0$.  Set $p=8\varepsilon$ in
Theorem~\ref{thm:sphere-parameterized}.  Its dimension condition becomes
$d\geq C\log n\,\log^2(1/\varepsilon)$ after changing a universal constant,
and its conclusion gives minimum out-degree $\Omega(1/\varepsilon)$ and
$\Omega(n/\varepsilon)$ arcs.
\end{sloppypar}
\end{proof}

\subsection{One sample for the moderate-error regime}

The spherical construction is in fact uniform in $\eps$: one random sample in
polylogarithmic dimension simultaneously witnesses the optimal
$\Omega(1/\eps)$ lower bound throughout the moderate-error regime.

\begin{corollary}
\label{cor:simultaneous-sphere-curve}
There are universal constants $c,C,\varepsilon_0>0$ such that, if
$d\geq C\log^3 n$, then with probability at least $1-n^{-1}$ one sample of
$n$ independent uniform points on $\mathbb S^{d-1}$ simultaneously has the
following property for every
\[
 C\sqrt{\frac{\log n}{n}}\leq\varepsilon\leq\varepsilon_0:
\]
every $(1-\varepsilon)$-almost navigable graph on the sample has minimum
out-degree at least $c/\varepsilon$ and at least $cn/\varepsilon$ arcs.
\end{corollary}

Thus, a single generic dataset captures the optimal behavior throughout the moderate-error regime.
Moreover, the required
dimension is only $O(\log^3 n)$, rather than polynomial in $n$.

The proof is deferred to Appendix~\ref{app:corsimul}.

\section{The Zarankiewicz correspondence}
\label{sec:zarankiewicz}

The spherical construction obtained a navigation lower bound by producing many
large near-target sets with little pairwise overlap.  This is not a phenomenon
specific to random geometry.  The same combinatorial structure appears
naturally in a classical problem from extremal graph theory.

For the Zarankiewicz problem, the tradeoff is between the number of edges
and the maximum codegree (Section~\ref{sec:prelim-codegree}).

The next theorem gives the bridge from such graphs to Euclidean navigation.
It says that the closed neighborhoods of \emph{every} finite graph can be
realized simultaneously as equal-radius Euclidean balls.
Figure~\ref{fig:codegree-lens} illustrates this correspondence, with common
neighbors represented by points in the intersection of two neighborhood balls.

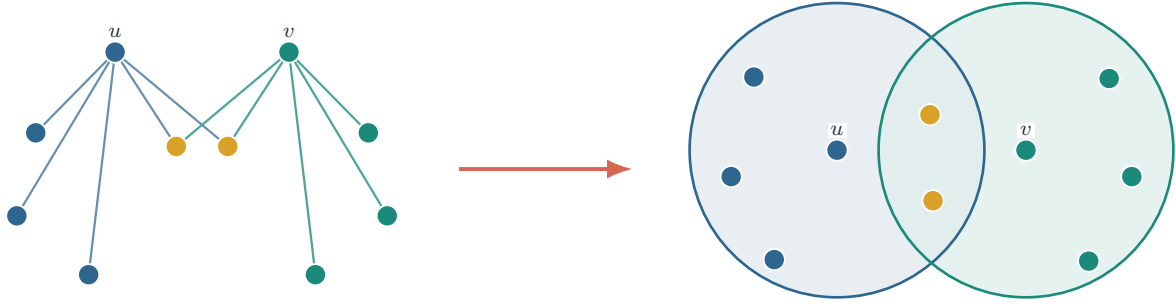
\begin{figure}[ht]
  \centering
  \resizebox{\linewidth}{!}{%
  \begin{tikzpicture}[x=1cm,y=1cm]

    \begin{scope}[shift={(-4.55,0)}]
      \coordinate (u)  at (-1.15,1.55);
      \coordinate (v)  at ( 1.15,1.55);
      \coordinate (s1) at (-.34,.30);
      \coordinate (s2) at ( .34,.30);

      \coordinate (a1) at (-2.20,.48);
      \coordinate (a2) at (-2.45,-.62);
      \coordinate (a3) at (-1.50,-1.40);

      \coordinate (b1) at ( 2.20,.48);
      \coordinate (b2) at ( 2.45,-.62);
      \coordinate (b3) at ( 1.50,-1.40);

      \foreach \p in {s1,s2,a1,a2,a3}
        \draw[draw=geBlue!72,line width=.85pt] (u)--(\p);

      \foreach \p in {s1,s2,b1,b2,b3}
        \draw[draw=geTeal!78,line width=.85pt] (v)--(\p);

      \node[ge blue vertex,
            label={[ge tiny label]above:$u$}] at (u) {};
      \node[ge teal vertex,
            label={[ge tiny label]above:$v$}] at (v) {};

      \foreach \p in {a1,a2,a3}
        \node[ge blue vertex] at (\p) {};

      \foreach \p in {b1,b2,b3}
        \node[ge teal vertex] at (\p) {};

      \node[ge gold vertex] at (s1) {};
      \node[ge gold vertex] at (s2) {};
    \end{scope}

    \draw[ge arrow] (-1.15,0) -- (1.15,0);

    \begin{scope}
      \def\ballradius{1.95}

      \coordinate (u) at (3.85,.25);
      \coordinate (v) at (6.35,.25);

      \begin{scope}[on background layer]
        \fill[geBlue!14,opacity=.78]
          (u) circle (\ballradius);
        \fill[geTeal!14,opacity=.78]
          (v) circle (\ballradius);
      \end{scope}

      \draw[draw=geBlue,line width=1pt]
        (u) circle (\ballradius);
      \draw[draw=geTeal,line width=1pt]
        (v) circle (\ballradius);

      \node[ge blue vertex,
            label={[ge tiny label]above:$u$}] at (u) {};
      \node[ge teal vertex,
            label={[ge tiny label]above:$v$}] at (v) {};

      \node[ge blue vertex] at (2.75, 1.22) {};
      \node[ge blue vertex] at (2.45,-.10) {};
      \node[ge blue vertex] at (3.02,-1.20) {};

      \node[ge teal vertex] at (7.45, 1.20) {};
      \node[ge teal vertex] at (7.75,-.10) {};
      \node[ge teal vertex] at (7.18,-1.22) {};

      \node[ge gold vertex] at (5.08, .72) {};
      \node[ge gold vertex] at (5.12,-.42) {};
    \end{scope}

  \end{tikzpicture}%
  }

\caption{Geometric realization of small codegree.
Common neighbors lie in the intersection of two neighborhood balls.}

  \Description{At left, two nonadjacent graph vertices each have three
  neighbors of their own and two shared neighbors. At right, each vertex is
  the center of an equal-radius ball. The two shared points lie in the
  overlap, and the other neighbors lie in their respective balls outside
  the overlap.}
  \label{fig:codegree-lens}
\end{figure}

\begin{theorem}
\label{thm:zarankiewicz-main}
Let $H$ be any simple graph on $n$ vertices.  There is a generic set
$P_H\subset\R^{n-1}$ and a radius $r>0$ such that, after identifying
$V(H)$ with $P_H$,

$$
 N_H[t]
 =
 \{x\in P_H:\|x-t\|_2\leq r\}
 \qquad\text{for every }t\in V(H).
$$

Moreover,

$$
 \max_{u\neq v}|N_H[u]\cap N_H[v]|
 \leq
 \codeg(H)+2.
$$

Consequently, for every $n$ and $\eps$,
\begin{equation}
\AN(n,\eps)
\geq
\max_{t\in\mathbb Z_{\geq0}}
\frac{2\Zar(n,t)-\eps n(n-1)}
{n(t+2)}.
\label{eq:zar-transference}
\end{equation}
Equivalently, any universal upper bound
$\AN(n,\eps)\leq A(n,\eps)$ implies

$$
 \Zar(n,t)
 \leq
 \frac{
 \eps n(n-1)+n(t+2)A(n,\eps)
 }{2}.
$$

\end{theorem}

\paragraph{Proof overview.}
A regular simplex on $n$ vertices has the same distance between every pair of
distinct points.  Starting from this configuration, we perturb the distances
slightly: pairs corresponding to edges of $H$ are made a little shorter,
while nonedges remain a little longer.  We then add much smaller, distinct
perturbations to remove all distance ties.

Because the perturbation is sufficiently small, the resulting squared-distance
matrix is still a Euclidean distance matrix, this follows from Schoenberg's
criterion~\cite{schoenberg1935remarks}.  Hence it is realized by $n$
affinely independent points in $\R^{n-1}$.  Choosing a radius strictly between
the edge-distance band and the nonedge-distance band gives

$$
 x\in N_H[t]
 \quad\Longleftrightarrow\quad
 \|x-t\|_2\leq r.
$$

Thus the combinatorial neighborhoods of $H$ become genuine metric balls in a
generic Euclidean dataset.

Global charging with
$B_t=N_H[t]$ then gives
\[
 |E(G)|\geq\frac{2|E(H)|-\eps n(n-1)}{\codeg(H)+2},
\]
which yields~\eqref{eq:zar-transference} after optimizing over $H$.  The
reverse inequality is a rearrangement.

\subsection{The common-ball realization}
\label{app:zar-proof}

The construction is a small perturbation of the squared-distance matrix of a
regular simplex: edge pairs are made slightly shorter than nonedge pairs,
while a still smaller perturbation makes all pairwise distances distinct.
\begin{proof}[Proof of the realization statement in
Theorem~\ref{thm:zarankiewicz-main}]
Let $A$ be the adjacency matrix of $H$. Choose a symmetric zero-diagonal
matrix $W$ whose entries $W_{ij}$, indexed by unordered pairs $\{i,j\}$,
are mutually distinct and lie in $(0,1)$. For $\alpha>0$, put
\[
M=\alpha(A+W/2),
\qquad
S_{ii}=0,
\qquad
S_{ij}=1-M_{ij}\quad(i\neq j).
\]
Thus $S$ is our candidate squared-distance matrix.

Choose $\alpha>0$ sufficiently small that
\[
\|M\|_{\mathrm{op}}<\frac12
\]
and every $S_{ij}$ is positive. Since
\[
S=\mathbf 1\mathbf 1^T-I-M,
\]
for every $z\perp\mathbf 1$ we have
\[
z^TSz
=
-\|z\|^2-z^TMz
\le
-\|z\|^2+\|M\|_{\mathrm{op}}\|z\|^2
<
-\frac12\|z\|^2.
\]
Hence $S$ is strictly conditionally negative definite.

Let
\[
\Pi=I-\frac1n\mathbf 1\mathbf 1^T,
\]
the orthogonal projection onto $\mathbf 1^\perp$, and define
\[
B=-\frac12\Pi S\Pi.
\]
By Schoenberg's Euclidean-distance-matrix criterion~\cite{schoenberg1935remarks}, $S$ is a squared
Euclidean distance matrix whenever $B\succeq0$, in that case $B$ is the
Gram matrix of the centered realization. For every $x$,
\[
x^TBx
=
-\frac12(\Pi x)^TS(\Pi x)
\ge
\frac14\|\Pi x\|^2.
\]
Therefore $B\succeq0$, with equality only when $\Pi x=0$. Thus
\[
\ker B=\operatorname{span}\{\mathbf 1\},
\qquad
\operatorname{rank}(B)=n-1.
\]
Consequently $S$ is the squared-distance matrix of $n$ affinely
independent points in $\mathbb R^{n-1}$.

It remains to verify the desired distance ordering. If $ij\in E(H)$, then
\[
M_{ij}\in(\alpha,3\alpha/2),
\]
whereas if $ij\notin E(H)$, then
\[
M_{ij}\in(0,\alpha/2).
\]
Hence
\[
S_{ij}<1-\alpha
\quad\text{for edges},
\qquad
S_{ij}>1-\alpha/2
\quad\text{for nonedges}.
\]
Choose
\[
r^2\in(1-\alpha,1-\alpha/2).
\]
Then
\[
ij\in E(H)
\quad\Longleftrightarrow\quad
d(i,j)<r,
\]
so for every vertex $t$,
\[
N_H[t]=\{x:d(x,t)\le r\}.
\]
Since the entries $W_{ij}$ are mutually distinct, all unordered pairwise
distances are distinct as well.

Finally, for distinct $u,v$,
\[
|N_H[u]\cap N_H[v]|
\le
|N_H(u)\cap N_H(v)|+2
\le
\operatorname{codeg}(H)+2,
\]
which proves the stated pair-intersection bound.
\end{proof}

\begin{corollary}
\label{cor:low-codegree-obstruction}
For every simple $n$-vertex graph $H$, there is a generic Euclidean
$n$-point set such that every $(1-\eps)$-almost navigable graph $G$ on it
satisfies
\[
 |E(G)|\geq
 \frac{2|E(H)|-\eps n(n-1)}{\codeg(H)+2}.
\]
\end{corollary}

\begin{proof}
Apply the global charging theorem with $B_t=N_H[t]$.  The sum of
$|B_t|-1$ is $2|E(H)|$, and Theorem~\ref{thm:zarankiewicz-main} bounds the pair load.
\end{proof}

\subsection{The transference inequality}

\begin{proof}[Proof of the transference statements in
Theorem~\ref{thm:zarankiewicz-main}]
Fix $t\geq0$ and let $H$ be extremal for $\Zar(n,t)$.  Then
$
|E(H)|=\Zar(n,t)
$
and
$
\codeg(H)\leq t.
$
The low-codegree obstruction in Corollary~\ref{cor:low-codegree-obstruction} gives

$$
|E(G)|
\geq
\frac{2\Zar(n,t)-\eps n(n-1)}{t+2},
$$

and therefore

$$
\AN(n,\eps)
\geq
\frac{2\Zar(n,t)-\eps n(n-1)}{n(t+2)}.
$$

Maximizing over $t$ proves the forward inequality.  Conversely, if
$\AN(n,\eps)\leq A(n,\eps)$, then

$$
2\Zar(n,t)-\eps n(n-1)
\leq n(t+2)A(n,\eps),
$$

hence

$$
\Zar(n,t)
\leq
\frac{\eps n(n-1)+n(t+2)A(n,\eps)}{2}.
$$

\end{proof}

The reverse direction explains why the navigation law has the same square-root
shape as the classical extremal bound.

\begin{corollary}
\label{cor:kst-recovery}
For $t\geq1$,
\[
                         \Zar(n,t)=O(\sqrt t\,n^{3/2}),
\]
and $\Zar(n,0)=\lfloor n/2\rfloor$.
\end{corollary}

To see the first bound, substitute the known navigation upper bound
$A(n,\eps)=O(1/\eps)$ into the reverse inequality.  This gives

$$
 \Zar(n,t)
 \leq
 O\!\left(
 \eps n^2+\frac{tn}{\eps}
 \right),
$$

up to lower-order terms.  Balancing the two contributions by taking
$\eps\asymp\sqrt{t/n}$ yields

$$
 \Zar(n,t)=O(\sqrt t n^{3/2}).
$$

\begin{proof}
When $t=0$, no vertex can have two distinct neighbors, so an extremal graph
is a matching.  For $t\geq1$, substitute $A(n,\eps)\leq C/\eps$ into
the rearranged inequality and take
\[
                  \eps=\min\{1,\sqrt{C(t+2)/n}\}.
\]
When this value is below one, the two
terms balance at $O(\sqrt{t+2}\,n^{3/2})=O(\sqrt t\,n^{3/2})$, otherwise the trivial
$O(n^2)$ bound is of the displayed order.
\end{proof}

This shows that the
$1/\eps$ law for almost navigability and the
$\sqrt tn^{3/2}$ law for bounded-codegree graphs are two manifestations of
the same counting principle.

We therefore recover the classical K\H{o}v'ari--S'os--Tur'an theorem by a
different route, using almost navigability as the intermediate object.

The correspondence works in both directions.  Navigation upper bounds give
Zarankiewicz upper bounds, while dense bounded-codegree graphs give hard
Euclidean instances for navigation.  

\section{F\"uredi's construction}
\label{sec:furedi}

The spherical construction gives the optimal $\Omega(1/\eps)$ lower bound
when $\eps$ is not too small.  As $\eps$ approaches zero, however, the
relevant scale changes: the general upper bound
$
 \AN(n,\eps)
 \leq
 O\!\left(\min\left\{\frac1\eps,\sqrt n\right\}\right)
$
saturates at $\sqrt n$.  To prove a matching lower bound in this regime, we
need a different source of large, low-overlap neighborhoods.

The Zarankiewicz correspondence from the previous section suggests exactly
where to look.  We would like a graph in which every vertex has many
neighbors, but every two vertices have only a few common neighbors.  Dense
$K_{2,t+1}$-free graphs have precisely these properties.  F\"uredi's
finite-field construction~\cite{furedi1996asymptotics} gives an especially
useful family: its degrees are of order $q$, while its codegrees are at most
$t$.  After realizing its closed neighborhoods as Euclidean balls, the
charging theorem therefore predicts a navigation lower bound of order
$q/t$.

There is one additional issue.  The general realization theorem from
Section~\ref{sec:zarankiewicz} places an arbitrary $N$-vertex graph in
$\R^{N-1}$.  For the present application we want a substantially
lower-dimensional Euclidean witness.  F\"uredi's graphs have enough spectral
structure to compress the realization while preserving the separation
between edges and nonedges.  This gives the following theorem.

\begin{theorem}
\label{thm:furedi-main}
Let $q$ be a prime power, let $t\mid(q-1)$, and put

$$
 N=\frac{q^2-1}{t}.
$$

There is a generic $N$-point Euclidean dataset in dimension
$O(q\log N)$ such that every $(1-\eps)$-almost navigable graph $G$
on the dataset satisfies
\begin{equation}
\min_{u}\outdeg_G(u)
\geq
\frac{\bigl(q-1-\eps(N-1)\bigr)_+}{t+2}.
\label{eq:furedi-sourcewise}
\end{equation}
\end{theorem}

\paragraph{Proof overview.}

F\"uredi's graph has minimum degree at least $q-1$ and maximum codegree at
most $t$.  After realization, each closed neighborhood becomes a common-radius
Euclidean ball containing at least $q$ points, while any two such balls
intersect in at most $t+2$ points.  Thus each source has at least $q-1$
nearby targets to serve, of which at most $\eps(N-1)$ may fail, while a
single outgoing edge can be reused for at most $t+2$ targets.  Corollary~
\ref{cor:local-common-ball} therefore gives
for every source $u$

$$
 \outdeg_G(u)
 \geq
 \frac{\bigl(q-1-\eps(N-1)\bigr)_+}{t+2}\, .
$$

To obtain a low-dimensional realization, we use the additional spectral
structure of F\"uredi's graph: its least adjacency eigenvalue is at least
$-\sqrt q-1$.  A suitable shifted adjacency matrix can therefore be factored
to give a Euclidean realization in which edges and nonedges occupy separated
distance ranges.  A Johnson--Lindenstrauss projection reduces the dimension
to $O(q\log N)$ while preserving this separation~\cite{frankl1988johnson},
a final small perturbation removes distance ties without changing the realized
neighborhoods.

\subsection{Spectral compression of common-ball realizations}
\label{app:spectral-compression-proof}

\begin{lemma}
\label{lem:spectral-compression}
Let $H$ be a simple graph on $m\geq2$ vertices with adjacency matrix $A_H$.
If $\lambda_{\min}(A_H)\geq-L$ for some $L\geq2$, then there is a generic
point set $P_H\subset\R^D$ and a common radius $r>0$ such that
\[
 N_H[t]=\{x\in P_H:d(x,t)\leq r\}
 \qquad\text{for every }t,
\]
where $D=O(L^2\log m)$.
\end{lemma}

The general realization uses dimension $m-1$, but a spectral bound on $A_H$ allows a much lower-dimensional realization via Johnson-Lindenstrauss~\cite{frankl1988johnson}.

\begin{proof}
Let $A=A_H$.  Since $\lambda_{\min}(A)\geq-L$, the matrix $A+LI$ is positive
semidefinite.  Choose vectors $x_1,\ldots,x_m$ having this matrix as their
Gram matrix.  Thus $\|x_i\|_2^2=L$, and for distinct $i,j$,
\[
 \|x_i-x_j\|_2^2
 =2L-2A_{ij}
 =
 \begin{cases}
  2L-2,&ij\in E(H),\\
  2L,&ij\notin E(H).
 \end{cases}
\]

Set $\eta=1/(8L)$.  The Johnson--Lindenstrauss lemma gives a linear map into
$\R^D$, with
\[
                              D=O(\eta^{-2}\log m)
                               =O(L^2\log m),
\]
such that every squared pairwise distance is multiplied by a factor in
$[1-\eta,1+\eta]$.  Write $y_i$ for the projected points.  Every edge then
has squared length at most $(1+\eta)(2L-2)$, whereas every nonedge has
squared length at least $(1-\eta)2L$.  The latter quantity exceeds the former,
since
\[
 (1-\eta)2L-(1+\eta)(2L-2)
 =2-\eta(4L-2)>0.
\]
Choose $r^2$ strictly between these two bands.  Then
\[
                         ij\in E(H)
                         \quad\Longleftrightarrow\quad
                         \|y_i-y_j\|_2<r.
\]

The two distance bands are separated by a positive margin, so this graph is
unchanged under sufficiently small perturbations of the points.  Inside any
such open neighborhood, configurations with a repeated unordered pairwise
distance lie in a finite union of proper algebraic hypersurfaces.  Choose a
configuration outside this union.  It is generic, still realizes $H$ at the
common radius $r$, and therefore satisfies
\[
                         N_H[t]=\{x:d(x,t)\leq r\}
\]
for every center $t$.
\end{proof}

\subsection{F\"uredi's finite-field construction}

We apply the preceding machinery to F\"uredi's finite-field graphs, using their degree, codegree, and spectral bounds together with Siegel-Walfisz formula to obtain the complete phase uniformly in the required parameters.

\begin{lemma}
\label{lem:furedi-graph}
Let $q$ be a prime power and let $t\mid(q-1)$.  There is a simple graph
$F(q,t)$ with
\[
       |V(F(q,t))|=\frac{q^2-1}{t},\qquad
       |E(F(q,t))|\geq\frac{(q^2-1)(q-1)}{2t},
\]
and
\[
                         \codeg(F(q,t))\leq t.
\]
If $A_{q,t}$ is its adjacency matrix, then also
\[
                         \lambda_{\min}(A_{q,t})
                         \geq-\sqrt q-1.
\]
In particular, $F(q,t)$ is $K_{2,t+1}$-free.
\end{lemma}

The construction and the spectral identity are proved in
Appendix~\ref{app:furedi-graph-construction}.

\begin{proof}[Proof of Theorem~\ref{thm:furedi-main}]
Apply Lemma~\ref{lem:spectral-compression} to $F(q,t)$ with
$L=\sqrt q+1$.  Lemma~\ref{lem:furedi-graph} gives
\[
                         D=O((\sqrt q+1)^2\log N)
                          =O(q\log N),
\]
and the resulting common-radius balls are exactly the closed neighborhoods
of $F(q,t)$.  Every such neighborhood has at least $q$ points, because the
minimum graph degree is at least $q-1$.  Moreover, for distinct $u,v$,
\[
                         |N[u]\cap N[v]|\leq t+2.
\]
Corollary~\ref{cor:local-common-ball} therefore yields, at every source,
\[
 \outdeg(u)
 \geq
 \frac{\bigl(q-1-\eps(N-1)\bigr)_+}{t+2},
\]
which is~\eqref{eq:furedi-sourcewise}.
\end{proof}

\subsection{The complete phase}

For arbitrary $n$, choose $t\asymp\max\{1,\eps^2n\}$ and a prime
$q\equiv1\pmod t$ with $q=\Theta(\sqrt{nt})$.  A uniform
Siegel--Walfisz estimate supplies this choice for the required range
$t=O(\log n)$~\cite[Corollary~11.21]{montgomeryvaughan2007}.  Padding the
resulting graph and applying global charging gives the small-error part of the
next corollary, Theorem~\ref{thm:sphere-main} supplies the remaining range.

\begin{lemma}
\label{lem:uniform-zar-lower}
For every fixed $B>0$ there are constants $c>0$ and $n_B$ such that, for
$n\geq n_B$ and every integer $1\leq t\leq B\log n$,
\[
                         \Zar(n,t)\geq c\sqrt t\,n^{3/2}.
\]
Moreover, the witnessing $n$-vertex graph has a generic common-radius
Euclidean realization in dimension
\[
                          O(\sqrt{nt}\log n).
\]
The constant $c$ can be chosen absolute, only $n_B$ depends on $B$.
\end{lemma}

The parameter selection and the padding argument are given in the proof
of Lemma~\ref{lem:uniform-zar-lower} in Appendix~\ref{app:uniform-zar-proof}.

\begin{corollary}
\label{cor:complete-phase}
There are universal constants $C_0,\eps_0>0$ such that, for every
sufficiently large $n$ and every $0\leq\eps\leq\eps_0$,

$$
 \AN(n,\eps)
 = \Theta\left(\min\left\{\frac1\eps,\sqrt n\right\}\right)
\, .
$$

Moreover, the lower bound is witnessed by a generic $n$-point Euclidean
dataset in dimension
\begin{equation}
D=
\begin{cases}
O\left(\log n\,\log^2(1/\eps)\right),
& \eps\geq C_0\sqrt{\log n/n},\\[1.5mm]
O\left( (\sqrt n+\eps n)\log n\right),
& 0\leq\eps<C_0\sqrt{\log n/n}.
\end{cases}
\label{eq:phase-dimension}
\end{equation}
In particular, when $\eps=0$, there are generic
$n$-point datasets in $\R^{O(\sqrt n\log n)}$ for which every fully navigable
graph has average out-degree $\Omega(\sqrt n)$, or equivalently
$\Omega(n^{3/2})$ edges.
\end{corollary}

\begin{proof}[Proof of Corollary~\ref{cor:complete-phase}]
The upper bound is~\eqref{eq:known-upper}.  Let $C_0$ be the constant in the
lower endpoint of Theorem~\ref{thm:sphere-main}, and decrease $\eps_0$ if
necessary so that the sphere theorem applies throughout
$[C_0\sqrt{\log n/n},\eps_0]$.

First suppose
\[
                       \eps\geq C_0\sqrt{\frac{\log n}{n}}.
\]
Theorem~\ref{thm:sphere-main} gives
$\AN(n,\eps)=\Omega(1/\eps)$ on a sample that is generic with probability
one, in dimension $O(\log n\log^2(1/\eps))$.  In this range
$1/\eps\leq\sqrt n$ for all sufficiently large $n$, so this is the desired
$\Omega(\min\{1/\eps,\sqrt n\})$ bound.

It remains to consider
\[
                       0\leq\eps<C_0\sqrt{\frac{\log n}{n}}.
\]
Let $c_0>0$ be the constant from Lemma~\ref{lem:uniform-zar-lower}, and choose
a fixed $K\geq c_0^{-2}$.  Put
\[
                         t=\max\{1,\lceil K\eps^2n\rceil\}.
\]
Then $t\leq1+KC_0^2\log n$, so the lemma applies (after increasing its fixed
constant $B$).  Choose its witnessing graph $H$, so
\[
                         |E(H)|\geq c_0\sqrt t\,n^{3/2},
                         \qquad \codeg(H)\leq t,
\]
and let $P_H$ be the supplied common-radius realization.
Because $t\geq K\eps^2n$,
\[
                     \eps n^2\leq c_0\sqrt t\,n^{3/2}.
\]
Applying global charging directly to these balls and using $t+2\leq3t$ gives,
for every almost navigable graph $G$ on $P_H$,
\[
 \frac{|E(G)|}{n}
 \geq
 \frac{2c_0\sqrt t\,n^{3/2}-\eps n(n-1)}{n(t+2)}
 \geq \frac{c_0}{3}\sqrt{\frac nt}.
\]
Finally, $t\leq1+K\eps^2n$.  Writing $z=\eps\sqrt n$, we have
\[
 \sqrt{\frac nt}
 \geq\frac{\sqrt n}{\sqrt{1+Kz^2}}
 \geq\frac1{\sqrt{K+1}}
       \min\left\{\sqrt n,\frac1\eps\right\},
\]
where the last inequality follows by considering $z\leq1$ and $z\geq1$
separately.  This proves the required lower bound.

Finally, the same lemma gives
\[
 D=O(\sqrt{nt}\log n)
  =O\!\left(\sqrt{n(1+K\eps^2n)}\,\log n\right)
  =O((\sqrt n+\eps n)\log n).
\]
Since $\eps<C_0\sqrt{\log n/n}$ in the present case, this is
$O(\sqrt n\log^{3/2}n)=o(n)$.  At $\eps=0$ we have $t=1$, and hence
$D=O(\sqrt n\log n)$.  Together with the sphere-regime dimension bound this
proves~\eqref{eq:phase-dimension}.
\end{proof}

The corollary identifies the complete worst-case tradeoff.  For moderate
error, allowing more failures reduces the required degree proportionally to
$1/\eps$, once $\eps$ falls below the order of $1/\sqrt n$, this improvement
can no longer continue, and the cost saturates at $\Theta(\sqrt n)$.  The
fully navigable case $\eps=0$ is therefore not a separate phenomenon but the
endpoint of the same phase transition.

Appendix~\ref{app:furedi-proof} gives a stronger sourcewise statement
at the exact algebraic sizes $N=(q^2-1)/t$.

This phenomenon is not unique to Füredi’s graphs: projective incidence graphs have the same degree–codegree structure and give analogous sourcewise bounds, see Appendix~\ref{app:projective-family-proof}.

\section{Discussion}
\label{sec:discussion-main}

The main technical theme of the paper is that navigation lower bounds can be
built from low-overlap neighborhood systems.  Random spherical caps provide
such systems directly in Euclidean space, while our Euclidean realization
theorem transfers arbitrary bounded-codegree graphs into the same geometric
framework.  This reveals a direct connection with the Zarankiewicz problem and
allows algebraic extremal constructions, in particular F\"uredi's graphs, to
produce sharp navigation lower bounds.

\begin{sloppypar}
The main remaining gap is dimensional.  The spherical construction is
polylogarithmic-dimensional, but the algebraic witnesses require substantially
larger dimension in the small-error regime.  It remains open whether the full
$\Theta(\min\left({1/\eps,\sqrt n}\right))$ phase, and in particular the
$\Omega(\sqrt n)$ fully navigable endpoint, can be realized in
polylogarithmic or constant dimension.
\end{sloppypar}

\clearpage
\bibliographystyle{ACM-Reference-Format}
\bibliography{references}

\clearpage
\appendix
\section{Proofs and supplementary results}
\label{app:lower-bounds}

\subsection{Relative spherical-cap estimates}
\label{app:sphere-proof}

Dimension-dependent cap intersections have been analyzed in idealized
graph-search models~\cite{prokhorenkova2020graph}.  We need a uniform
\emph{relative} estimate at a cap mass that may vanish with $n$, the Gaussian
step below is a two-dimensional specialization of the multivariate Mills-ratio
method of Savage~\cite{savage1962mills}, included in full to track constants.

\begin{lemma}
\label{lem:sphere-marginal}
There are absolute constants $c,C,c_0>0$ such that, for $d\geq8$ and
$1\leq a\leq\sqrt d/4$,
\begin{equation}
       p_d(a)\geq
       c\exp\!\left(-C\frac{a^4}{d}\right)\overline\Phi(a),
\end{equation}
where $\overline\Phi(a)=\Pr(Z\geq a)$ for $Z\sim N(0,1)$.
Moreover, for every $d\geq8$ and every $1\leq a<\sqrt d$,
\begin{equation}
                        p_d(a)\leq e^{-c_0a^2}.
\end{equation}
\end{lemma}

\begin{proof}
The density of $Y=\sqrt d\,\langle X,u\rangle$ is
\begin{equation}
 f_{d,1}(y)=
 \frac{\Gamma(d/2)}{\sqrt{\pi d}\,\Gamma((d-1)/2)}
 \left(1-\frac{y^2}{d}\right)^{(d-3)/2}
 \mathbf 1\{|y|<\sqrt d\}.
\end{equation}
The gamma-ratio prefactor in the one-coordinate density formula is between two
absolute constant multiples of $(2\pi)^{-1/2}$.  This follows directly from
Stirling's inequalities
\[
 \sqrt{2\pi}\,z^{z+1/2}e^{-z}
 \leq\Gamma(z+1)
 \leq\sqrt{2\pi}\,z^{z+1/2}e^{-z+1/(12z)},
 \qquad z\geq1.
\]

For $0\leq s\leq1/2$,
\begin{equation}
             -s-\frac{s^2}{2(1-s)}
             \leq\log(1-s)\leq-s.
\end{equation}
Apply the lower inequality with $s=y^2/d$.  For
$y\in[a,a+1/a]$, we have $y\leq2a$ and hence $y^2/d\leq1/4$.  Comparing
the one-coordinate density formula with the standard Gaussian density
$\phi(y)$ gives
\begin{equation}
             f_{d,1}(y)
             \geq c\exp\!\left(-C\frac{a^4}{d}\right)\phi(y).
\end{equation}
Also, for $a\geq1$,
\begin{equation}
 \int_a^{a+1/a}\phi(y)\,dy
 \geq \frac{1}{a}\phi(a+1/a)
 \geq e^{-3/2}\frac{\phi(a)}{a}
 \geq e^{-3/2}\overline\Phi(a),
\end{equation}
where the last inequality is the elementary Mills upper bound
$\overline\Phi(a)\leq\phi(a)/a$.  Integrating the local density lower bound on
this interval proves the marginal comparison.

For the upper bound, we use the spherical-cap Chernoff argument.
If $G_1,\ldots,G_d$ are independent standard Gaussians,
then $G/\|G\|_2$ is uniform on the sphere.  Equivalently, direct integration
of the one-coordinate density formula and
$\log(1-s)\leq-s$ gives, for every $1\leq a<\sqrt d$,
\[
 p_d(a)
 \leq C\sqrt d\int_{a/\sqrt d}^1
            e^{-(d-3)s^2/2}\,ds
 \leq C e^{-(d-3)a^2/(2d)}.
\]
For $a$ above a sufficiently large absolute constant, the prefactor $C$ is
absorbed by decreasing the coefficient of $a^2$ in the exponent.  Below that
constant, $p_d(a)\leq1/2$ supplies the same conclusion after decreasing that
coefficient once more.  This proves the spherical subgaussian estimate for a
universal $c_0>0$.
\end{proof}

\begin{lemma}
\label{lem:two-cap}
There are absolute constants $C_0,C_1>0$ such that the following holds.
Let $d\geq C_0$, $1\leq a\leq\sqrt d/4$, and let $u,v$ be unit vectors with
$|\langle u,v\rangle|\leq1/2$.  Writing
$\rho=\langle u,v\rangle$ and $p=p_d(a)$,
\begin{equation}
 \Pr\bigl(X\in\operatorname{Cap}(u,a)
                 \cap\operatorname{Cap}(v,a)\bigr)
 \leq C_1 p^2
       \exp\!\left(C_1|\rho|a^2+C_1\frac{a^4}{d}\right).
\end{equation}
\end{lemma}

\begin{proof}
Rotate coordinates so that
$u=e_1$ and $v=\rho e_1+\sqrt{1-\rho^2}\,e_2$.  The joint density of
$(Y,W)=\sqrt d(X_1,X_2)$ is
\begin{equation}
 f_{d,2}(y,w)=
 \frac{d-2}{2\pi d}
 \left(1-\frac{y^2+w^2}{d}\right)^{(d-4)/2}
 \mathbf 1\{y^2+w^2<d\}.
\end{equation}
We first compare this density globally to two independent standard Gaussians.
For $s=y^2+w^2<d$, the logarithm of the density ratio, apart from the
nonpositive term $\log(1-2/d)$, is
\begin{equation}
                    g_d(s)=\frac{s}{2}
                    +\frac{d-4}{2}\log(1-s/d).
\end{equation}
Its derivative is
\[
                    g_d'(s)=\frac{4-s}{2(d-s)}.
\]
Thus $g_d$ is maximized at $s=4$ (or at an endpoint when $d\leq4$), and
$\sup_s g_d(s)\leq2$.  Therefore
\begin{equation}
                         f_{d,2}(y,w)\leq e^2\phi(y)\phi(w)
                         \qquad(d>4).
\end{equation}

Under the Gaussian density in the global comparison above, the variables
\[
                   Y,\qquad Z=\rho Y+\sqrt{1-\rho^2}\,W
\]
are standard Gaussians of correlation $\rho$.  We next bound their joint
upper tail without losing the Gaussian Mills prefactor.  Their density at
$(y,z)$ is
\[
 \frac{1}{2\pi\sqrt{1-\rho^2}}
 \exp\!\left[-Q_\rho(y,z)\right],
 \qquad
 Q_\rho(y,z)=\frac{y^2-2\rho yz+z^2}{2(1-\rho^2)}.
\]
For $y=a+s,z=a+t$ with $s,t\geq0$, convexity gives the supporting-plane
bound
\begin{equation}
 Q_\rho(a+s,a+t)
 \geq Q_\rho(a,a)+\frac{a}{1+\rho}(s+t)
 =\frac{a^2}{1+\rho}+\frac{a}{1+\rho}(s+t).
\end{equation}
Integration of the supporting-plane bound yields the explicit
Savage-type bound~\cite{savage1962mills}
\begin{equation}
 \Pr(Y\geq a,Z\geq a)
 \leq
 \frac{(1+\rho)^2}{2\pi a^2\sqrt{1-\rho^2}}
 \exp\!\left(-\frac{a^2}{1+\rho}\right).
\end{equation}
For $a\geq1$, the elementary Mills lower bound
\begin{equation}
                 \overline\Phi(a)\geq
                 \frac{a}{1+a^2}\phi(a)
                 \geq\frac{1}{2a}\phi(a)
\end{equation}
and $|\rho|\leq1/2$ imply
\begin{equation}
 \Pr(Y\geq a,Z\geq a)
 \leq C\,\overline\Phi(a)^2
       \exp(C|\rho|a^2).
\end{equation}
Indeed, the ratio of the exponential terms in the Gaussian joint-tail bound
and $\overline\Phi(a)^2$ is at most
$\exp(a^2\rho/(1+\rho))$ when $\rho\geq0$, and is at most one when
$\rho<0$.

Combining the global density comparison and Gaussian relative-tail estimate,
then using
Lemma~\ref{lem:sphere-marginal} in the form
\[
             \overline\Phi(a)
             \leq C p\exp(Ca^4/d),
\]
proves the two-cap bound.
\end{proof}

The two error terms in the two-cap bound reflect the two losses used in the
proof.  Positive center correlation changes the Gaussian large-deviation
exponent by $\Theta(\rho a^2)$, while our comparison between spherical and
Gaussian marginal tails incurs an $\exp(O(a^4/d))$ correction in the
moderate-deviation range.

\subsection{Proof of Lemma~\ref{lem:uniform-incidence-codegree}}
\label{app:sphere-incidence}

\begin{proof}
We first control all center correlations.  The one-coordinate density formula,
together with
$\log(1-s)\leq-s$, implies that for a fixed unit $u$ and
$0<r\leq1/2$,
\begin{equation}
 \Pr\bigl(|\langle X,u\rangle|>r\bigr)
 \leq C\exp\!\left(-\frac{(d-3)r^2}{2}\right)
\end{equation}
whenever $r\sqrt d\geq1$, the factor from integrating the Gaussian tail is
absorbed into $C$.  Set
\begin{equation}
                         r_n=4\sqrt{\frac{\log n}{d}}.
\end{equation}
Increasing $C_2$ makes $r_n\leq1/2$, and a union bound using the coherence
tail estimate gives
\begin{equation}
 \Pr\left(\max_{u\ne v}|\langle x_u,x_v\rangle|>r_n\right)
 \leq n^{-4}.
\end{equation}

For a fixed $u$, conditional on $x_u$, the indicators
$\mathbf1\{x_u\in B_t\}$, $t\ne u$, are independent Bernoulli variables of
mean $p$.  Hence $D_u\sim\operatorname{Bin}(n-1,p)$.  The multiplicative
Chernoff inequality
\begin{equation}
          \Pr(Z\leq\mu/2)\leq e^{-\mu/8}
          \qquad(Z\sim\operatorname{Bin}(N,q),\ \mu=Nq)
\end{equation}
and a union bound show that every $D_u\geq p(n-1)/2$ except with probability
at most $ne^{-p(n-1)/8}=o(n^{-3})$.

Fix now $u\ne v$ and condition on $x_u,x_v$.  Apart from the at most two
centers $t\in\{u,v\}$, the codegree $C_{uv}$ is binomial with $n-2$ trials
and success probability
\[
 q_{uv}=\Pr\bigl(X\in\operatorname{Cap}(x_u,a)
                        \cap\operatorname{Cap}(x_v,a)\bigr).
\]
On the event $|\langle x_u,x_v\rangle|\leq r_n$,
Lemma~\ref{lem:two-cap}, the upper bound on $a$, and the parameter range give
\begin{align}
 r_na^2
 &\leq4C_5\sqrt{\frac{\log n}{d}}\log(1/p)
 \leq \frac{4C_5}{\sqrt{C_2}},
 \\
 \frac{a^4}{d}
 &\leq\frac{C_5^2\log^2(1/p)}{d}
 \leq\frac{C_5^2}{C_2\log n}.
\end{align}
Thus, after fixing $C_2$, $q_{uv}\leq C_4p^2$ for a universal $C_4$.

We use the following additive binomial Chernoff bound, obtained by optimizing
the exponential moment:
\begin{equation}
 \Pr(Z\geq\mu+s)
 \leq\exp\!\left(-\frac{s^2}{2(\mu+s/3)}\right),
 \qquad Z\sim\operatorname{Bin}(N,q),\quad\mu=Nq.
\end{equation}
Since $p^2n\geq A^2\log n$, choose $A$ large enough that the binomial
Bernstein bound, with threshold $4C_4p^2n$, is at most
$n^{-6}$.  More explicitly, for each pair,
\[
\Pr\!\left(C_{uv}>4C_4p^2n+2,\
           |\langle x_u,x_v\rangle|\leq r_n\right)
\leq n^{-6},
\]
this conditions only on the fixed pair $(x_u,x_v)$ and does not condition
the other centers on the global coherence event.  A union bound over fewer
than $n^2/2$ pairs, followed by the uniform coherence event, proves that
$C_{uv}\leq4C_4p^2n+2\leq C_3p^2n$ for all pairs, except with probability
at most $2n^{-4}$.  After increasing $n_0$, the incidence failure probability
is at most $n^{-3}$, so the coherence, codegree, and incidence failures sum to
less than $n^{-2}$.  This proves the stated probability.
\end{proof}

\subsection{Proof of Corollary~\ref{cor:simultaneous-sphere-curve}}
\label{app:corsimul}

\begin{proof}[Proof of Corollary~\ref{cor:simultaneous-sphere-curve}]
Take a dyadic grid $\varepsilon_0,\varepsilon_0/2,\ldots$ down to the last
value at least $C\sqrt{\log n/n}$.  It has $O(\log n)$ values.  Increase the
dimension constant so that Theorem~\ref{thm:sphere-main} applies at every
grid value, and union bound its $n^{-2}$ failure probability over the grid.
For all sufficiently large $n$ the resulting failure probability is at most
$n^{-1}$.  Given an intermediate $\varepsilon$, choose a grid value
$\varepsilon'$ with
$\varepsilon\leq\varepsilon'<2\varepsilon$.  Every
$(1-\varepsilon)$-almost navigable graph is also
$(1-\varepsilon')$-almost navigable, so the grid guarantee gives minimum
out-degree at least $c/\varepsilon'\geq c/(2\varepsilon)$, summing over
sources gives the arc bound.
\end{proof}

\subsection{Proof of Lemma~\ref{lem:furedi-graph}}
\label{app:furedi-graph-construction}

\begin{proof}
This is F\"uredi's construction~\cite{furedi1996asymptotics}, we include the
short argument to keep track of its finite parameters.  Let
$\mathbb F=\mathbb F_q$ and let $H\leq\mathbb F^\times$ be the subgroup of
order $t$.  On $\mathbb F^2\setminus\{(0,0)\}$ identify
\[
                    (a,b)\sim(ha,hb)\qquad(h\in H),
\]
and write $\langle a,b\rangle$ for an equivalence class.  Every class has
exactly $t$ representatives, so there are $(q^2-1)/t$ vertices.  Two distinct
classes $\langle a,b\rangle$ and $\langle x,y\rangle$ are adjacent when
\[
                              ax+by\in H.
\]
The condition is independent of the chosen representatives.

Fix $\langle a,b\rangle$.  For each $h\in H$, the equation
$ax+by=h$ is an affine line containing $q$ points.  Taking all $h\in H$ gives
$qt$ representative solutions, and every equivalence class among them occurs
with all its $t$ representatives.  Thus there are $q$ candidate neighboring
classes, at most one is the vertex itself and is omitted because the graph is
simple.  Hence every degree is at least $q-1$, which gives the asserted edge
count.

Now take distinct vertices $\langle a_1,b_1\rangle$ and
$\langle a_2,b_2\rangle$.  A common neighbor $\langle x,y\rangle$ yields
$h_1,h_2\in H$ with
\[
 a_1x+b_1y=h_1,\qquad a_2x+b_2y=h_2.
\]
The coefficient matrix is nonsingular.  Indeed, if
$(a_1,b_1)=c(a_2,b_2)$, then the two displayed equations imply
$c=h_1h_2^{-1}\in H$, making the two original classes equal.  Therefore every
ordered pair $(h_1,h_2)\in H^2$ determines at most one solution $(x,y)$.
There are $t^2$ such pairs.  Multiplying a solution by any $g\in H$ multiplies
$(h_1,h_2)$ by the same $g$, so each common-neighbor class accounts for $t$
of these pairs.  There are consequently at most $t$ common neighbors.

For the spectral statement, let $\widetilde A$ be the adjacency matrix of the
same relation with possible diagonal entries retained, thus
\[
 \widetilde A_{\langle a,b\rangle,\langle x,y\rangle}=1
 \quad\Longleftrightarrow\quad
 ax+by\in H.
\]
Every row of $\widetilde A$ has exactly $q$ ones.  Partition the vertex set
according to the one-dimensional $\mathbb F$-subspace spanned by a
representative.  There are $q+1$ parts, each of size
\[
                              s=\frac{q-1}{t}.
\]
Let $C$ be the block-diagonal matrix whose $q+1$ diagonal blocks are the
$s\times s$ all-ones matrix, and let $J$ be the all-ones matrix of order
$(q^2-1)/t$.  The same two-equation count used above gives the exact identity
\begin{equation}
                          \widetilde A^2=qI+t(J-C).
 \label{eq:furedi-square-identity}
\end{equation}
Indeed, the diagonal entries are $q$, two distinct classes in the same
one-dimensional subspace have no common neighbor, and two classes in
different subspaces have exactly $t$ common neighbors.

The matrix $J-C$ has eigenvalue $qs=q(q-1)/t$ on the all-ones vector,
eigenvalue $-s$ on the $q$-dimensional space of vectors that are constant on
each part and sum to zero globally, and eigenvalue zero on the space of
vectors summing to zero inside every part.  Hence $\widetilde A^2$ has
eigenvalues $q^2$, $1$, and $q$.  Since
$\widetilde A\mathbf 1=q\mathbf 1$, every eigenvalue of $\widetilde A$ other
than the Perron eigenvalue $q$ has magnitude at most $\sqrt q$.  In
particular,
\[
                         \lambda_{\min}(\widetilde A)\geq-\sqrt q.
\]

The simple adjacency matrix is $A_{q,t}=\widetilde A-D$, where $D$ is a
diagonal $0$--$1$ matrix recording the possible loops.  Since
$0\preceq D\preceq I$, Weyl's inequality gives
\[
 \lambda_{\min}(A_{q,t})
 \geq\lambda_{\min}(\widetilde A)-1
 \geq-\sqrt q-1,
\]
as claimed.
\end{proof}

\subsection{Proof of Lemma~\ref{lem:uniform-zar-lower}}
\label{app:uniform-zar-proof}

\begin{proof}
Put
\[
                              x=\frac12\sqrt{nt}.
\]
We claim that, for all sufficiently large $n$ uniformly in
$1\leq t\leq B\log n$, there is a prime
\begin{equation}
                        x<q\leq2x,\qquad q\equiv1\pmod t.
\label{eq:prime-dyadic-ap}
\end{equation}
For $t=1$ this follows, for example, from Bertrand's postulate.  Suppose
$t\geq2$.  Since $\log x=\Theta(\log n)$ uniformly in the stated range, for
large enough $n$ we have $t\leq(\log x)^2$.  A standard Siegel-Walfisz form
of the prime number theorem in arithmetic progressions states that, for every
fixed $A>0$, uniformly for $m\leq(\log y)^A$ and $(a,m)=1$,
\begin{equation}
 \vartheta(y,m,a)=\frac{y}{\varphi(m)}
            +O_A\!\left(y\exp(-c_A\sqrt{\log y})\right),
\label{eq:siegel-walfisz-form}
\end{equation}
where $\vartheta(y,m,a)=\sum_{p\leq y,\ p\equiv a\, (\mathrm{mod}\,m)}\log p$
and $c_A>0$ is a constant~\cite[Corollary~11.21]{montgomeryvaughan2007}.
Apply this with $A=2$ at $y=x$ and $y=2x$.  Since
$t\leq B\log n=O_B(\log x)$,
\[
 \vartheta(2x,t,1)-\vartheta(x,t,1)
 =\frac{x}{\varphi(t)}+O\!\left(x\exp(-c\sqrt{\log x})\right)>0
\]
for all sufficiently large $n$, uniformly in the stated range of $t$.  Hence
there is a prime $q\equiv1\pmod t$ in $(x,2x]$, which
proves~\eqref{eq:prime-dyadic-ap}.

For this prime $q$, Lemma~\ref{lem:furedi-graph} applies because
$t\mid(q-1)$.  Moreover $q\leq2x=\sqrt{nt}$, so
\[
              \frac{q^2-1}{t}<n.
\]
Pad $F(q,t)$ with isolated vertices to obtain an $n$-vertex graph, neither its
edge count nor its maximum codegree increases.  Since $q>x$ and $q\to\infty$
uniformly, for sufficiently large $n$,
\[
 |E(F(q,t))|
 \geq\frac{(q^2-1)(q-1)}{2t}
 \geq\frac{q^3}{8t}
 \geq\frac1{64}\sqrt t\,n^{3/2}.
\]
Padding by isolated vertices only appends zero adjacency eigenvalues and
therefore preserves the lower bound from Lemma~\ref{lem:furedi-graph}: the
padded graph has least eigenvalue at least $-\sqrt q-1$.  Lemma~\ref{lem:spectral-compression}
gives a generic common-radius realization in dimension
\[
 O(q\log n)=O(\sqrt{nt}\log n),
\]
because $q\leq\sqrt{nt}$.  Thus the conclusion holds, for example, with
$c=1/64$, together with the asserted dimension bound.
\end{proof}

\subsection{Sourcewise endpoint on exact F\"uredi sizes}
\label{app:furedi-proof}

\begin{corollary}
\label{cor:furedi-sourcewise-endpoint}
For every prime power $q$, with $N=q^2-1$, there is a generic $N$-point set
in $\R^{O(\sqrt N\log N)}$ such that every fully navigable graph satisfies
\[
                         \min_u\outdeg_G(u)\geq\frac{q-1}{3}
                         =\Omega(\sqrt N).
\]
More generally, the same conclusion up to a factor of two holds whenever
$\eps(N-1)\leq(q-1)/2$.
\end{corollary}

\begin{proof}
Apply Theorem~\ref{thm:furedi-main} with $t=1$.
\end{proof}

\subsection{Additional sourcewise constructions from projective incidence geometry}
\label{app:projective-family-proof}

Although projective incidence geometry is no longer needed to close the
worst-case asymptotics, the same charging framework yields a useful symmetric
family in which the bound applies uniformly at every source.  The case $m=2$
is the projective-plane specialization.  More generally, let $PG(m,q)$ be projective $m$-space over $\mathbb F_q$, and
consider its bipartite point--hyperplane incidence graph.  Standard
finite-projective-space counts~\cite{dembowski1968finite} give the following
parameters.  Each side has
\[
v=1+q+\cdots+q^m
\]
vertices, every vertex has degree
\[
k=1+q+\cdots+q^{m-1},
\]
and two distinct vertices on the same side have exactly
\[
\lambda=1+q+\cdots+q^{m-2}
\]
common neighbors.  For a point--hyperplane pair, the corresponding closed
neighborhoods intersect in at most the two endpoints.

Thus every source has $k$ incidence neighbors, while the overlap relevant
to a single witnessing edge is at most $\max\{\lambda,2\}$.  Combining
these exact incidence parameters with the local charging bound and the
Euclidean realization yields the following sourcewise family.

\begin{theorem}
\label{thm:projective-family}
For every prime power $q$ and integer $m\geq2$, let
\[
 N=2(1+q+\cdots+q^m),\qquad
 k=1+q+\cdots+q^{m-1},\qquad
 \lambda=1+q+\cdots+q^{m-2}.
\]
There is a generic Euclidean $N$-point set such that, whenever
$\eps(N-1)\leq k/2$, every $(1-\eps)$-almost navigable graph satisfies
\[
 \min_u\outdeg_G(u)\geq\frac{k}{2\max\{\lambda,2\}}=\Omega(q).
\]
In particular, taking $\eps=k/[4(N-1)]$ gives $\eps=\Theta(1/q)$ and
minimum out-degree $\Omega(1/\eps)$.
\end{theorem}

\begin{proof}
Let $H$ be the bipartite point--hyperplane incidence graph of
$PG(m,q)$. Recall that each side has
\[
v=1+q+\cdots+q^m
\]
vertices, so
\[
N:=|V(H)|=2v=2(1+q+\cdots+q^m).
\]
Every vertex has degree
\[
k=1+q+\cdots+q^{m-1},
\]
and two distinct vertices on the same side have exactly
\[
\lambda=1+q+\cdots+q^{m-2}
\]
common neighbors.

Apply the realization statement in Theorem~\ref{thm:zarankiewicz-main} to $H$. We obtain a generic point set
$P_H\subset\mathbb R^{N-1}$ and a common radius $r>0$ such that,
for every $t\in P_H$,
\[
B_t:=\{x\in P_H:d(x,t)\le r\}=N_H[t].
\]
In particular, since $H$ is $k$-regular,
\[
|B_u|-1=k
\]
for every $u\in P_H$.

We next bound the pairwise intersections of these balls. Let
$u,v\in V(H)$ be distinct. If $u$ and $v$ lie on the same side of
the bipartition, then neither endpoint lies in the other's closed
neighborhood, and therefore
\[
|N_H[u]\cap N_H[v]|=\lambda.
\]
If $u$ and $v$ lie on opposite sides and are incident, then
\[
N_H[u]\cap N_H[v]=\{u,v\},
\]
while if they are nonincident their closed neighborhoods are
disjoint. Hence, for every source $u$,
\[
\max_{v\ne u}|B_u\cap B_v|
\le \max\{\lambda,2\}.
\]

Now let $G$ be any $(1-\varepsilon)$-almost navigable directed graph
on $P_H$. By Corollary~\ref{cor:local-common-ball}, for every source $u$,
\[
\deg_G^+(u)
\ge
\frac{|B_u|-1-\varepsilon(N-1)}
{\max\{1,\max_{v\ne u}|B_u\cap B_v|\}}.
\]
Using the bounds above gives
\[
\deg_G^+(u)
\ge
\frac{k-\varepsilon(N-1)}
{\max\{\lambda,2\}}.
\]
Under the hypothesis
\[
\varepsilon(N-1)\le \frac{k}{2},
\]
we conclude that
\[
\deg_G^+(u)
\ge
\frac{k}{2\max\{\lambda,2\}}
\]
for every $u\in P_H$. Thus
\[
\min_u \deg_G^+(u)
\ge
\frac{k}{2\max\{\lambda,2\}}.
\]

Finally, since $m\ge 2$,
\[
k=1+q+\cdots+q^{m-1}=\Theta(q^{m-1})
\]
and
\[
\lambda=1+q+\cdots+q^{m-2}=\Theta(q^{m-2}),
\]
so
\[
\frac{k}{\max\{\lambda,2\}}=\Theta(q).
\]
Therefore
\[
\min_u\deg_G^+(u)=\Omega(q).
\]

In particular, if
\[
\varepsilon=\frac{k}{4(N-1)},
\]
then $\varepsilon(N-1)=k/4\le k/2$, and since
\[
\frac{k}{N-1}=\Theta(1/q),
\]
we have
\[
\varepsilon=\Theta(1/q).
\]
Consequently,
\[
\min_u\deg_G^+(u)=\Omega(q)=\Omega(1/\varepsilon).
\]
\end{proof}

\begin{corollary}
\label{cor:projective-plane}
For every prime power $q$, with $N=2(q^2+q+1)$, there is a generic Euclidean
$N$-point set such that every fully navigable graph satisfies
\[
                         \min_u\outdeg_G(u)=\Omega(\sqrt N).
\]
\end{corollary}

\begin{proof}
Apply Theorem~\ref{thm:projective-family} with $m=2$ and $\eps=0$.
Then $k=q+1$ and $\lambda=1$.
\end{proof}

The incidence construction is classical, our corollary uses its symmetric
design identities as a sourcewise navigability obstruction.  The sizes and
error values are discrete, so its role is complementary to the uniform
average-degree phase theorem rather than necessary for interpolation.

\end{document}